\documentclass[a4paper,10pt]{article}
\usepackage{amsmath}
\usepackage{graphicx}%
\usepackage{multirow}%
\usepackage{amsmath,amssymb,amsfonts}%
\usepackage[title]{appendix}%
\usepackage{hyperref}%
\usepackage{xcolor}%
\usepackage{textcomp}%
\usepackage{manyfoot}%
\usepackage{booktabs}%
\usepackage{algorithm}%
\usepackage{algorithmicx}%
\usepackage{algpseudocode}%
\usepackage{listings}%
\usepackage{latexsym}
\usepackage{extarrows}
\usepackage{color}
\usepackage{graphicx}
\usepackage{mathrsfs}
\usepackage{array}
\usepackage{makecell}
\usepackage{geometry}
\usepackage{lettrine}
\usepackage{multicol}
\usepackage{indentfirst}
\usepackage{cite}
\usepackage{mathtools}
\usepackage{graphicx}
\usepackage{graphics}
\usepackage{subfigure}
\usepackage{caption}
\usepackage{booktabs}
\usepackage{multirow}
\usepackage{diagbox}
\usepackage{makecell}
\usepackage{geometry}
\usepackage{lettrine}
\usepackage{multicol}
\usepackage{indentfirst}
\usepackage{cite}
\usepackage{mathtools}
\usepackage{graphicx}
\usepackage{graphics}
\usepackage{subfigure}
\usepackage{caption}
\usepackage{booktabs}
\usepackage{multirow}
\usepackage{diagbox}
\usepackage{makecell}
\usepackage{amsthm}%
\usepackage{mathrsfs}%
\usepackage[title]{appendix}%
\usepackage{hyperref}%
\usepackage{xcolor}%
\usepackage{textcomp}%
\usepackage{manyfoot}%
\usepackage{booktabs}%
\usepackage{algorithm}%
\usepackage{algorithmicx}%
\usepackage{algpseudocode}%
\usepackage{listings}%
\usepackage{makecell}
\usepackage{float}
\makeatletter

\newcommand{\Rmnum}[1]{\expandafter\@slowromancap\romannumeral #1@}

\newtheorem{theorem}{Theorem}
\newtheorem{definition}{Definition}

\newtheorem{corollary}[theorem]{Corollary}
\newtheorem{lemma}[theorem]{Lemma}

\title{Constructions and List Decoding of Sum-Rank Metric Codes Based on Orthogonal Spaces over Finite Fields}
\author{Xuemei Liu \and Jiarong Zhang \and Gang Wang\textsuperscript{$^*$}}
\date{\small College of Science, Civil Aviation University of China, 300300, Tianjin, China. 
\\$^*$Corresponding author, E-mail: gwang06080923@mail.nankai.edu.cn}

\begin{document}

	\maketitle

\begin{abstract}
	Sum-rank metric codes, as a generalization of Hamming codes and rank metric codes, have important applications in fields such as multi-shot linear network coding, space-time coding and distributed storage systems. The purpose of this study is to construct sum-rank metric codes based on orthogonal spaces over finite fields, and calculate the list sizes outputted by different decoding algorithms. The following achievements have been obtained.

In this study, we construct a cyclic orthogonal group of order $q^n-1$ and an Abelian non-cyclic orthogonal group of order $(q^n-1)^2$ based on the companion matrices of primitive polynomials over finite fields. By selecting different subspace generating matrices, maximum rank distance (MRD) codes with parameters $(n \times {2n}, q^{2n}, n)_q$ and $(n \times {4n}, q^{4n}, n)_q$ are constructed respectively. Two methods for constructing sum-rank metric codes are proposed for the constructed MRD codes, and the list sizes outputted under the list decoding algorithm are calculated. Subsequently, the $[{\bf{n}},k,d]_{{q^n}/q}$-system is used to relate sum-rank metric codes to subspace designs. The list size of sum-rank metric codes under the list decoding algorithm is calculated based on subspace designs. This calculation method improves the decoding success rate compared with traditional methods.


 \medskip
\end{abstract}

\section{Introduction}
Coding theory is a branch of information theory whose origins can be dated back to the middle of the 20th century. With the development of communication technology, noise and interference inevitably occur during data transmission, leading to information loss or errors. To ensure the accuracy and integrity of information, Shannon\textsuperscript{\cite{ref0001}}introduced the concept of channel capacity and proved that information can be transmitted accurately over a noisy channel using an encoding scheme that adds redundant information. This provides a theoretical basis for future research on error-correcting codes.

Classical error-correcting codes can be divided into two main categories. One is Hamming codes\textsuperscript{\cite{ref1}}. They can automatically correct errors at the single bit level, which is an important tool for data transmission in early computer storage and communication systems. The other category is rank metric codes, such as Gabidulin codes, which can correct errors at the symbol level. Rank metric codes are used in the network coding and distributed storage systems. However, in modern communication networks, both types of error can occur concurrently during data transmission. To address this problem, by building on the concepts of Hamming codes and rank metric codes, Nóbrega et al.\textsuperscript{\cite{ref8}} presented a coding scheme that can correct both single bit and symbol level errors, namely sum-rank metric codes.

The sum-rank metric codes are generalizations of the Hamming codes and rank metric codes. In addition to previous research on Hamming and rank metric codes, sum-rank metric codes can be applied in  multishot network coding\textsuperscript{\cite{ref9, ref10, ref11}}, space-time coding\textsuperscript{\cite{ref12, ref13}} and distributed storage\textsuperscript{\cite{ref14, ref15}}. The properties of sum-rank metric codes have been extensively studied\textsuperscript{\cite{ref16, ref17, ref19}}. In 2010, Nóbrega et al.\textsuperscript{\cite{ref8}} firstly combined Hamming codes with rank metric codes to construct sum-rank metric codes. Let $\mathbb{F}_q$ be a finite field with $q$ elements, where $q$ denotes the power of the prime. Let $t, m_1, m_2,\cdots, m_t, n_1, n_2, \cdots, n_t$ be positive integers.
The space formed by the direct sum of $t$ matrix spaces is defined as follows:
\[\Pi  = \mathop  \oplus \limits_{i = 1}^t \mathbb{F}_q^{{m_i} \times {n_i}}.\]
Take any $X = ( X_1| X_2| \cdots| X_t )$, $Y = ( Y_1|  Y_2|  \cdots| Y_t)  \in \Pi $, where $X_i, Y_i \in \mathbb{F}_q^{{m_i} \times {n_i}}.$ Then we define sum-rank distance:
\begin{align}
d_{SR}:\Pi  \times \Pi  &\to N\nonumber \\
( {X,Y} ) &\mapsto \sum\limits_{i = 1}^t {{{\rm{rank}}}( X_i - Y_i)} .\nonumber
\end{align}
Let $\mathcal{C}$ be an $\mathbb{F}_q$-linear subspace of $\Pi$, endowed with the above-defined sum-rank distance. Then, $\mathcal{C}$ is called the sum-rank metric code with parameters $((m_1 \times n_1| m_2 \times n_2| \cdots | m_t \times n_t), {|\mathcal{C}|}, d)_q$. The sum-rank weight of element $X = ( X_1| X_2| \cdots| X_t ) \in \Pi$ is denoted by 
$$w_{SR}( X ) = \sum\limits_{i = 1}^t {{{\rm{rank}}}( X_i)}.$$ The minimum sum-rank distance is denoted by 
$$d= \min \left\{ {w_{SR}( X )\left| {X \in \mathcal{C},X \ne 0} \right.} \right\}.$$
When $m_i =n_i =1 (1\leq i \leq t)$, the sum-rank metric codes are Hamming codes. When $t = 1$, the sum-rank metric codes are rank metric codes.

Determining codes with the maximal number of codewords under a given minimum distance is a fundamental topic in coding theory. Byrne et al.\textsuperscript{\cite{ref16}} proved that the Singleton bound of the sum-rank metric codes also exists. Therefore, we call the sum-rank metric codes whose number of codewords reaches the Singleton bound the maximum sum-rank distance (MSRD) codes. As a result, much research on sum-rank metric codes focuses on constructing codes whose number of codewords approaches the Singleton bound. 


In 2020, Martínez-Peñas \textsuperscript{\cite{ref20}} constructed sum-rank Hamming codes under the sum-rank metric, as well as their dual codes, namely, sum-rank simplex codes. Among them, the sum-rank Hamming code with codeword sizes $m_1= m_2= \cdots =m_t$ and $n_1=n_2=\cdots=n_t$ is an MSRD code. Building on this work, in 2021, Byrne et al.\textsuperscript{\cite{ref16}} presented several methods for constructing MSRD codes with specific parameters using MDS and MRD codes. Although these methods can not provide matrices of arbitrary sizes within the codewords, they result in matrices with unequal numbers of columns. Subsequently, in 2023, Chen\textsuperscript{\cite{ref21}} used algebraic geometry codes to construct $\mathbb{F}_q$-linear MSRD codes with matrices of size $n_1 \times n_1, n_2 \times n_2, \cdots, n_t \times n_t $, satisfying the condition ${n_i}^2 \ge {n_{i+1}}^2+ \cdots+{n_t}^2 $ ($i= 1, 2, \cdots, t-1$). In addition, in 2024, Chen et al.\textsuperscript{\cite{ref22}} used cyclic and Goppa codes to construct $\mathbb{F}_2$-sum-rank metric codes with matrices of size $2 \times 2$. In the same year, Lao et al.\textsuperscript{\cite{ref23}} increased the number of codewords of the sum-rank metric codes while maintaining the minimum sum-rank distance, leading to the construction of linear MSRD codes over $\mathbb{F}_q$. Later, Martínez-Peñas\textsuperscript{\cite{ref24}} modified the existing MSRD codes in 2024 to explicitly obtain codes with non-uniform row and column numbers in the matrices of their codewords. Inspired by these articles, this paper constructs sum-rank metric codes from a new perspective.


The purpose of this study was to construct sum-rank metric codes and calculate the list size of the decoding algorithm under the actions of orthogonal groups over finite field. First, by leveraging the companion matrices of primitive polynomials, we construct a cyclic orthogonal group of order $q^n -1$ and an Abelian non-cyclic orthogonal group of order ${(q^n -1)}^2$, which play crucial roles in the subsequent construction of codes. Next, based on these groups, we determine the generating matrices of the subspace to construct new MRD codes with parameters $(n \times {2n}, q^{2n}, n)_q$ and $(n \times {4n}, q^{4n}, n)_q$. Then, we provide two different methods for constructing sum-rank metric codes based on MRD codes, and obtain sum-rank metric codes with parameters $[(4|\cdots|4),2(2t-1),2]_{{q^2}/q}$, $[(8|\cdots|8),4(2t-1),2]_{{q^2}/q}$, $[(n | \cdots|n | 2n), 2, t(n-1)+1]_{{q^n}/q}$ and $[(n|\cdots|n|4n),4,t(n-3)+3]_{{q^n}/q}$. Finally, we calculate the list size of the sum-rank metric codes under the general list decoding algorithm and the list decoding algorithm based on subspace designs. Through a comparison with the general list decoding algorithm, it was found that the list decoding algorithm based on subspace designs can significantly improve the decoding success rate.

The remainder of this paper is organized as follows. In Section \ref{sec2}, we provide basic definitions of orthogonal spaces over finite fields, sum-rank metric codes and subspace designs. In Section \ref{sec3}, we describe the construction of new MRD codes and sum-rank metric codes. In Section \ref{sec4}, we calculate the list size of sum-rank metric codes using the general list decoding algorithm and the list decoding algorithm based on subspace designs. Section \ref{sec5} concludes the paper.

\section{Preliminaries}\label{sec2}
Let $\mathbb{F}_q$ be a finite field of odd characteristics and define $\mathbb{F}^*_q$ as the set of all nonzero elements in $\mathbb{F}_q$, then $\mathbb{F}^*_q$ is a cyclic group of order $q-1$ with respect to multiplication. Let $\mathbb{F}^{*2}_q$ denote the subgroup formed by the squares of the elements in $\mathbb{F}^*_q$, that is, $\mathbb{F}^{*2}_q=\{b| b=c^2, c\in \mathbb{F}^*_q\}$. Taking a non-square element $z$ in $\mathbb{F}^*_q$, we have $\mathbb{F}^*_q=\mathbb{F}^{*2}_q\cup z\mathbb{F}^{*2}_q$. 

Let $S$ be an $n\times n$ nonsingular symmetric matrix over $\mathbb{F}_q$. An $n\times n$ matrix $T$ over $\mathbb{F}_q$ is said to be orthogonal with respect to $S$ if 
\[
TS^tT = S.
\]

Clearly, $n\times n$ orthogonal matrices with respect to a nonsingular symmetric matrix $S$ are nonsingular and form a group with respect to matrix multiplication, called the orthogonal group of degree $n$ with respect to $S$ over $\mathbb{F}_q$ and denoted by $O_n(\mathbb{F}_q,S)$. The index of $S$ is also called the index of the orthogonal group $O_n(\mathbb{F}_q,S)$. If $S_1$ and $S_2$ are two cogredient $n\times n$ nonsingular symmetric matrices over $\mathbb{F}_q$, then $O_n(\mathbb{F}_q,S_1)$ and $O_n(\mathbb{F}_q,S_2)$ are isomorphic. Hence we only discuss the four orthogonal groups with respect to the following four $n\times n$ nonsingular symmetric matrices:

\[
S_{2\nu}=\begin{pmatrix}
0_{\nu}&I_{\nu}\\
I_{\nu}&0_{\nu}
\end{pmatrix}, 
S_{2\nu + 1,1}=\begin{pmatrix}
0_{\nu}&I_{\nu}\\
I_{\nu}&0_{\nu}&\\
&&1
\end{pmatrix}, 
\]

\[
S_{2\nu+1,z}=\begin{pmatrix}
0_{\nu}&I_{\nu}\\
I_{\nu}&0_{\nu}&\\
&&z
\end{pmatrix}, 
S_{2\nu + 2}=\begin{pmatrix}
0_{\nu}&I_{\nu}\\
I_{\nu}&0_{\nu}&\\
&&1
\end{pmatrix},
\]
where $n = 2\nu, 2\nu + 1, 2\nu+1$, and $2\nu + 2$, respectively. To cover these four cases, we introduce the notation $S_{2\nu+\delta,\Delta}$, where $\nu$ is its index and $\Delta$ denotes its definite part, that is,

\[
\Delta=\begin{cases}
\phi, &\text{if }\delta = 0,\\
(1)\text{ or }(z),&\text{if }\delta = 1,\\
{\left( {\begin{array}{*{20}{c}}
1&{}\\
{}&{ - z}
\end{array}} \right),}&\text{if }\delta = 2.
\end{cases}
\]

The orthogonal group of degree $2\nu+\delta$ with respect to $S_{2\nu+\delta,\Delta}$ over $\mathbb{F}_q$ is denoted by $O_{2\nu+\delta,\Delta}(\mathbb{F}_q)$. Clearly, $O_{2\nu + 1}(\mathbb{F}_q,S_{2\nu+1,1})=O_{2\nu + 1}(\mathbb{F}_q,zS_{2\nu+1,z})$. Because $zS_{2\nu+1,z}$ and $S_{2\nu+1,1}$ are cogredient, the groups $O_{2\nu + 1}(\mathbb{F}_q,zS_{2\nu+1,z})$ and $O_{2\nu+1,1}(\mathbb{F}_q)$ are isomorphic. Therefore, only three types of orthogonal group must be considered. For simplicity, we sometimes write $S$ as $S_{2\nu+\delta,\Delta}$.

\begin{definition} 
(Orthogonal Space, see \cite{ref25}) \label{Df1} 
There is an action of $O_{2\nu+\delta,\Delta}(\mathbb{F}_q)$ defined as follows:
\begin{align}
\mathbb{F}_q^{(2\nu+\delta)}\times O_{2\nu+\delta,\Delta}(\mathbb{F}_q)&\to\mathbb{F}_q^{(2\nu+\delta)}  \nonumber \\
((x_1,x_2,\ldots,x_{2\nu+\delta}),T)&\to(x_1,x_2,\ldots,x_{2\nu+\delta})T. \nonumber 
\end{align} 
The elements of $O_{2\nu+\delta,\Delta}(\mathbb{F}_q)$ are called orthogonal transformations with respect to $S$. The vector space $\mathbb{F}_q^{(2\nu+\delta)}$ together with the above group action of the orthogonal group $O_{2\nu+\delta,\Delta}(\mathbb{F}_q)$ is called the $(2\nu+\delta)$-dimensional orthogonal space over $\mathbb{F}_q$ with respect to $S$.
\end{definition} 

First, we examine how the vector subspaces of $\mathbb{F}_q^{(2\nu+\delta)}$ are partitioned into orbits under $O_{2\nu+\delta,\Delta}(\mathbb{F}_q)$.

\begin{lemma} (See \cite{ref25}) \label{Le1} 
Let $P$ be an $m$-dimensional vector subspace of $\mathbb{F}_q^{(2\nu+\delta)}$. Then $PS^tP$ is an $m\times m$ symmetric matrix. $PS^tP$ is cogredient with one of the following four normal forms:

\[
M(m,2s,s)=\begin{pmatrix}
0_{s}&I_{s}\\
I_{s}&0_{s}&\\
&&0_{m - 2s}
\end{pmatrix},
\]

\[
M(m,2s + 1,s,1)=\begin{pmatrix}
0_{s}&I_{s}\\
I_{s}&0_{s}&\\
&&1&\\
&&&0_{m - 2s-1}
\end{pmatrix},
\]

\[
M(m,2s + 1,s,z)=\begin{pmatrix}
0_{s}&I_{s}\\
I_{s}&0_{s}&\\
&&z&\\
&&&&0_{m - 2s-1}
\end{pmatrix},
\]

and 

\[
M(m,2s + 2,s)=\begin{pmatrix}
0_{s}&I_{s}\\
I_{s}&0_{s}&\\
&&1&\\
&&&-z&\\
&&&&0_{m - 2s-2}
\end{pmatrix}.
\]
\end{lemma} 

We use the symbol $M(m,2s+\gamma,s,\Gamma)$ to represent any one of these four normal forms, where $s$ is its index, $\gamma = 0,1$, or $2$, and $\Gamma$ represents the definite part in these normal forms. In particular, when $\gamma = 1$, we have
$\Gamma=(1)$ or $(z)$; when $\gamma = 0$ or $2$,
 
\[
\Gamma=\phi \text{ or } \begin{pmatrix}
1&\\
&-z
\end{pmatrix},
\]
respectively, and may be omitted. If $PS^tP$ is cogredient to $M(m,2s+\gamma,s,\Gamma)$, then $P$ is called a subspace of type $(m,2s+\gamma,s,\Gamma)$ with respect to $S$ in $\mathbb{F}_q^{(2\nu+\delta)}$. 

Let $P$ be an $m$-dimensional subspace of $\mathbb{F}_q^{(2\nu+\delta)}$, we use the symbol $P^{\perp}$ to denote the set of vectors which are orthogonal to every vector of $P$ with respect to $S$, i.e.,
\[
P^{\perp} = \{y\in\mathbb{F}_q^{(2\nu+\delta)}\ |\ yS^tx = 0\text{ for all }x\in P\}.
\]
Clearly, $P^{\perp}$ is a $(2\nu+\delta - m)$-dimensional subspace of $P$ and is called the dual subspace of $P$ with respect to $S$.

Let $m$, $n$ be two positive integers and $\mathbb{F}_{q}^{m\times n}$ be the set consisting of all matrices of type $m \times n$ over $\mathbb{F}_q$. Next we introduce rank metric codes and their properties.

\begin{definition} (Rank Metric Code) \label{Df4}
For matrices $A, B \in \mathbb{F}_{q}^{m\times n}$, the rank distance is defined as 
\[d_{R}(A,B)={{\rm{rank}}{(A - B)}},\]
where ${\text{rank}}(\cdot)$ represents the rank of matrix. A nonempty subset $\mathcal{C}$ of $\mathbb{F}_{q}^{m\times n}$ is called rank metric code (RMC). The minimum rank distance $d_R(\mathcal{C})$ of $\mathcal{C}$ is defined as:
\[d_R(\mathcal{C})={\min}\{ d_{R}(A,B) | A, B \in \mathcal{C}, A\ne B   \}.\] 
\end{definition}

For a rank metric code $\mathcal{C}\subseteq \mathbb{F}_{q}^{m\times n}$,  $\mathcal{C}$ is called a $(m\times n, |\mathcal{C}|, d)_q$ code if the cardinality and minimum rank distance of $\mathcal{C}$ are $|\mathcal{C}|$ and $d$, respectively. Code $\mathcal{C}$ is called an $[m\times n, k, d]_q$ code if the dimensions of $\mathcal{C}$ are $k$ in $\mathbb{F}_{q}^{m\times n}$. 


\begin{lemma}
(Singleton Bound in Rank Metric Code, see \cite{ref7}) \label{Th1}
For every rank metric code $\mathcal{C} \subseteq \mathbb{F}_{q}^{m\times n}$ with $d_R(\mathcal{C})=d$, 
\[ |\mathcal{C}| \le q^{\max\{m,n\} (\min\{m,n\}-d+1) }.\]
\end{lemma}

Code $\mathcal{C}$ that achieves this bound is called the maximum rank distance (MRD) code. Next, we introduce the vector sum-rank metric codes.

We set ${\bf n} = \left( {{n_1},{n_2}, \cdots ,{n_t}} \right) \in \mathbb{N}^t$ for denoting an ordered tuples with ${n_1}\ge{n_2}\ge  \cdots \ge{n_t}$ and $N={n_1}+{n_2}+ \cdots +{n_t}$.
We used the following compact notation for the direct sum of the vector spaces:
\[\mathbb{F}_{{q^m}}^{\bf n} = \mathop  \oplus \limits_{i = 1}^t \mathbb{F}_{{q^m}}^{{n_i}}.\]

Let a vector ${a} = \left( {{a_1},{a_2}, \cdots ,{a_n}} \right) \in \mathbb{F}_{{q^m}}^{n}$, the rank of vector ${a}$ is the dimension of the vector space generated over $\mathbb{F}_{q}$ by its entries, i.e., ${\rm{rank}}_q\left(  { a}  \right) = {\dim _q}( {{{\left\langle {{a_1}, \cdots ,{a_n}} \right\rangle }_{_{q}}}} )$. The sum-rank weight of an element $ {\bf X} = \left( {{x_1}, {x_2}, \cdots , {x_t}} \right) \in \mathbb{F}_{{q^m}}^{\bf n}$ is 
\[{w}_{SR}( {\bf X} ) = \sum\limits_{i = 1}^t {{\rm{rank}}_q\left( {{x_i}} \right)}.\]

\begin{definition} (Vector Sum-rank Metric Code, see \cite{ref310})  \label{Df5}
A (vector) sum-rank metric code $\mathcal{C}$ is an $\mathbb{F}_{{q^m}}$-subspace of $\mathbb{F}_{{q^m}}^{\bf n}$ endowed with the sum-rank distance defined as 
\[d_{SR}({\bf{X}},{\bf{Y}}) = {w}_{SR}({\bf{X}} - {\bf{Y}}) = \sum\limits_{i = 1}^t {{{\rm{rank}}_q}} ({x_i} - {y_i}),\]
where $ {\bf X} = \left( {{x_1}, {x_2}, \cdots , {x_t}} \right),   {\bf Y} = \left( {{y_1}, {y_2}, \cdots , {y_t}} \right) \in \mathbb{F}_{{q^m}}^{\bf n}$.
\end{definition} 

Let $\mathcal{C} \subseteq \mathbb{F}_{{q^m}}^{\bf n}$ be a vector sum-rank metric code, where $\mathcal{C}$ is an $[ {\bf n},k,d ]_{{q^m}/q}$ code if $k$ is the $\mathbb{F}_{{q^m}}$-dimension of $\mathcal{C}$ and $d$ is its minimum distance, that is
\[d_{SR}(\mathcal{C})={\min}\{ d_{SR}({\bf X},{\bf Y}) | {\bf X}, {\bf Y} \in \mathcal{C}, {\bf X}\ne {\bf Y}   \}.\]




\begin{lemma}(Singleton Bound in Vector Sum-rank Metric Code, see \cite{ref310})   \label{Th2}  Let $\mathcal{C} \subseteq \mathbb{F}_{{q^m}}^{\bf n}$ be an $[ {\bf n},k,d ]_{{q^m}/q}$-vector sum-rank metric code. Let $j$ and $\delta$ be unique integers that satisfy
$$d - 1 = \sum\limits_{i = 1}^{j - 1} {\min \left\{ {m,{n_i}} \right\}}  + \delta ,   0 \le \delta  \le \min \left\{ {m,{n_i}} \right\} - 1.$$
Then
\begin{equation}
\left| \mathcal{C} \right| \le {q^{m\sum\limits_{i = j}^t {{n_i}}  - \max \left\{ {m,{n_j}} \right\}\delta }}.\nonumber
\end{equation}
In particular, if $m \ge n_i$, then $d - 1 = \sum\limits_{i = 1}^{j - 1} {{n_i}}  + \delta $ with $0 \le \delta  \le {n_j} - 1$ and
\[\left| \mathcal{C} \right| \le {q^{m\left( {N - d + 1} \right)}}.\]
If $n=n_1=\cdots =n_t \ge m$, then $d - 1 = m\left( {j - 1} \right) + \delta $ with $0 \le \delta  \le m - 1$ and
\[\left| \mathcal{C} \right| \le {q^{n\left( {tm - d + 1} \right)}}.\]
\end{lemma}

A $[{\bf n}, k, d]_{{q^m}/q}$-code is called the maximum sum-rank distance (MSRD) code if its size attains the above bound.

Define the coverage radius of $[{\bf n}, k, d]_{{q^m}/q}$-sum-rank metric code ${\mathcal{C}}$ as \[{\tau _{SR}}( {\mathcal{C}} ) = \mathop {\max }\limits_{x \in {\mathbb{F}^{\bf n}_{q^m}} } \mathop {\min }\limits_{c \in {\mathcal{C}}} \{ {w{_{SR}}( {x - c} )} \}.\] We constructed a sphere centered at any codeword $x$ in code ${\mathcal{C}}$ with radius ${\tau _{SR}}$: \[{\mathcal{B}_{SR}}( {x,{\tau _{SR}}( {\mathcal{C}} )} ) = \{ {y \in {\mathbb{F}^{\bf n}_{q^m}} |w{_{SR}}( {x - y} ) \le {\tau _{SR}}( {\mathcal{C}} )} \}.\]
The sphere completely covers the entire space ${\mathbb{F}^{\bf n}_{q^m}}$, where radius ${\tau _{SR}}( {\mathcal{C}} ) $ is the smallest radius that satisfies the condition.

\begin{definition}(See {\cite{ref340}}) \label{Df6}
The sum-rank metric code ${\mathcal{C}} \subseteq  {\mathbb{F}^{\bf n}_{q^m}} $ is said to be $(\tau_{SR} ( {\mathcal{C}} ), L)$-list decodable if each sphere ${\mathcal{B}_{SR}}( x, \tau_{SR} ( {\mathcal{C}} )  )$ contains at most $L$ codewords in code ${\mathcal{C}}$, where $L$ is the list size of ${\mathcal{C}}$ under the list decoding algorithm, and $\tau_{SR}({\mathcal{C}} )$ is denoted as the radius of error correction of ${\mathcal{C}}$. 
\end{definition}

\begin{definition} (See \cite{ref310} )\label{Df7}
An ordered set ${\cal H} = ( {{H_1},{H_2}, \cdots ,{H_{t}}} )$, where ${H_i}$ is an $\mathbb{F}_{{q}}$-subspace of $\mathbb{F}_{{q}}^{mk}$ for any $i \in [t]$, is called an $(s, A)_q$-subspace design in $\mathbb{F}_{{q^{m}}}^k$ if $\dim_{q^m}{ \langle {H_1},{H_2}, \cdots ,{H_{t }}\rangle_{q^m}} \ge s$ and for every ${\mathbb{F}_{{q^m}}}$-subspace $W \subseteq \mathbb{F}_{{q^{m}}}^k$ of dimension $s$($s < k$),
\[\sum\limits_{i = 1}^{t } {{{\dim }_{q}}\left( {{H_i} \cap W} \right) \le A} .\]
\end{definition} 

Note that the dimension of $W$ is defined as $\mathbb{F}_{q^m}$, and the dimension of the intersection ${{H_i} \cap W}$ is defined as $\mathbb{F}_q$. Furthermore, the relevant properties of the subspace design can be found at \cite{ref311, ref318, ref319}. 


\section{Constructing sum-rank metric codes by MRD codes based on orthogonal spaces over finite fields }\label{sec3}

In this section, we construct the MRD codes  using cyclic orthogonal groups and Abelian non-cyclic orthogonal groups. Based on these MRD codes, we present two methods for constructing sum-rank metric codes.

\subsection{Construction of  MRD codes based on cyclic orthogonal groups}

Let 
\[f(x)= x^n +a_{n-1}x^{n-1}+ \cdots +a_1x+a_0 \]
be a primitive polynomial of degree $n$ over $\mathbb{F}_q$.
The companion matrix of $f(x)$ is
\[A_g= \left( {\begin{array}{*{20}{c}}
0&0& \cdots &0&{ - {a_0}}\\
1&0& \cdots &0&{ - {a_1}}\\
0&1& \cdots &0&{ - {a_2}}\\
 \vdots & \vdots &{}& \vdots & \vdots \\
0&0& \cdots &1&{ - {a_{n - 1}}}
\end{array}} \right)_{n }.\]
The companion matrix has the following properties:\\
(1) $f(A_g)= 0_n$, $f(x)$ is both the minimal polynomial and the characteristic polynomial of $A_g$;\\
(2) $\{{ 0}_n, I_n, A_g, A_g^2, \cdots, A_g^{q^n-2}\}$ is a finite field of order $q^n$ with respect to matrix addition and multiplication;\\
(3) ${\rm{ord}}(A_g)=q^n-1$;\\
(4) $I_n+ A_g+ A_g^2+ \cdots+ A_g^{q^n-2} = 0_n$;\\
(5) $I_n+ A_g+ A_g^2+ \cdots+ A_g^{a} \ne 0_n, {\rm{rank}}(I_n+ A_g+ A_g^2+ \cdots+ A_g^{a})=n $, where $1 \le a <{q^n-1}.$

Next, we constructed MRD codes using cyclic and Abelian non-cyclic orthogonal groups.

For $S_{2n}=\left( {\begin{array}{*{20}{c}}
0_n&{I_n}\\
{I_n}&0_n
\end{array}} \right)$, we constructed the matrix
\[A_1 = {\left( {\begin{array}{*{20}{c}}
{A_g}&{}\\
{}&{{}^tA_g^{ - 1}}
\end{array}} \right)_{2n}},\]
where $|A_1|=q^n-1$.
Since
\[\left( {\begin{array}{*{20}{c}}
{{A_g}}&{}\\
{}&{{}^tA_g^{ - 1}}
\end{array}} \right){S_{2n}}{}{{^t}\left( {\begin{array}{*{20}{c}}
{{A_g}}&{}\\
{}&{{}^tA_g^{ - 1}}
\end{array}} \right) }= \left( {\begin{array}{*{20}{c}}
{}&{{A_g}}\\
{{}^tA_g^{ - 1}}&{}
\end{array}} \right)\left( {\begin{array}{*{20}{c}}
{{}^t{A_g}}&{}\\
{}&{A_g^{ - 1}}
\end{array}} \right) = {S_{2n}},\]
$G_1 = \langle A_1\rangle$ is a cyclic orthogonal group of order $q^n -1$ over $\mathbb{F}_q$ with respect to $S_{2n}$, namely, $G_1$ is a cyclic subgroup of $O_{2n}(\mathbb{F}_q, S_{2n})$.


Selecting the generating matrices of subspace
$$W_1=(I_n \quad 0_n)_{n \times 2n} \in \mathbb{F}_q^{n \times {2n}} ,$$
$$W_{11}=( 0_n \quad I_n)_{n \times 2n} \in \mathbb{F}_q^{n \times {2n}} .$$
Then $W_1$ and $W_{11}$ are two matrices of rank $n$.

\begin{lemma}\label{Le3}
Let 
$${\cal C}_1 = {W_1}{G_1} \cup \{ {{0_{n \times 2n}}} \},$$
then code ${\cal C}_1$ is a linear-rank metric code with parameters $(n \times {2n}, q^n, n)_q$. Similarly, code ${\cal C}_{11} $ is a linear-rank metric code with parameters $(n \times {2n}, q^n, n)_q$.
 \end{lemma}
 
\begin{proof}

Firstly, we prove that ${\cal C}_1$ is a linear space over $\mathbb{F}_q$. 
\begin{align}
{W_1}{G_1} &= \left\{ {{W_1}{A^i_1}\;|0\le i\le {q^n} - 2 } \right\} \nonumber \\
&= \left\{ {({I_n}\quad 0_n){{\left( {\begin{array}{*{20}{c}}
{{A_g}}&{}  \\
{}&{^tA_g^{ - 1}}
\end{array}} \right)}^i}\;|0\le i\le {q^n} - 2 } \right\}  \nonumber \\
&= \left\{ {{{({A^i_g}\quad 0_n)}_{n \times 2n}}\;|0\le i\le {q^n} - 2 } \right\}.\nonumber
\end{align}
For any $W_1{A^i_1}$, $W_1{A^j_1} \in {\cal C}_1$ and $a$, $b \in \mathbb{F}_q$, 
we get
$$aW_1{A^i_1}+bW_1{A^j_1} = (a{A^i_g}+b{A^j_g}\quad 0_n).$$
If $a{A^i_g}+b{A^j_g} =0_n$, then $aW_1{A^i_1}+bW_1{A^j_1}=0_{n \times 2n} \in {\cal C}_1$;
If $a{A^i_g}+b{A^j _g}\ne 0_n$, then there exists an integer $k$ such that $a{A^i_g}+b{A^j_g}= {A^k_g}$. So
$$ aW_1{A^i_1}+bW_1{A^j_1}= W_1{A^k _1}\in {\cal C}_1.$$

Next, we calculated the number of codewords. An integer $i$ exists such that $W_1{A^i_1} = ({A^i_g}\quad 0_n) =(I_n\quad 0_n) = W_1$, and we obtain ${A^i_g} = I_n $, $ i \equiv 0({\rm{mod}}\  q^n-1)$. Furthermore, $\text{Stab}_{G_1}{W_1} =\{ I_{2n} \}$ and
$$|{\cal C}_1| = |W_1G_1|+1 = \frac{|G_1|}{|\text{Stab}_{G_1}{W_1}|}+1 =q^n.$$

Finally, we calculate the minimum rank distance of code ${\cal C}_1$. For any codeword $W_1{A^i_1} \in {\cal C}_1$, 
$$d_R({\cal C}_1) =d_R(W_1{A^i_1}, 0_{n \times {2n}}) = {\rm{rank}}(W_1{A^i_1}) ={\rm{rank}}(W_1) =n.$$
 Therefore, ${\cal C}_1$  is a linear-rank metric code with parameters $(n \times {2n}, q^n, n)_q$. Namely, ${\cal C}_{11} = {W_{11}}{G_1} \cup \{ {{0_{n \times 2n}}} \}$ is also a linear-rank metric code with parameters $(n \times {2n}, q^n, n)_q$.
 \end{proof}

\begin{theorem}\label{Th4}
Let $${\cal C}^{(1)} = {\cal C}_1 +{\cal C}_{11} = \{ c_1 +c_{11} | c_1 \in {\cal C}_1, c_{11} \in {\cal C}_{11} \},$$ 
then ${\cal C}^{(1)}$ is a linear MRD code with parameters $(n \times {2n}, q^{2n}, n)_q$.
\end{theorem}

\begin{proof} Obviously, ${\cal C}_1 \cap {\cal C}_{11} =\{ {\bf 0} \}$, then $|{\cal C}^{(1)}| = |{\cal C}_1| |{\cal C}_{11}| = q^{2n}$. In addition, $${{\cal C}^{(1)}} = \left\{ {({A^i_g}\quad {{{(^t}{A_g}^{ - 1})}^j}) \left|  {i,j = 0,1, \cdots ,{q^n} - 2} \right.} \right\} \cup {{\cal C}_1} \cup {{\cal C}_{11}},$$ 
$${\rm{rank}}({A^i_g}\quad {(^t}{A_g}^{ - 1})^j)={\rm{rank}}({A^i_g})=n,$$we get
$${d_R}({{\cal C}^{(1)}}) = n.$$

Thus, ${\cal C}^{(1)}$ is a linear MRD code with the parameters $(n\times {2n}, q^{2n}, n)_q$.
\end{proof}


Similarly, based on $S_{2n+1,1}=\left( {\begin{array}{*{20}{c}}
{\begin{array}{*{20}{c}}
0_n&{{I_n}}\\
{{I_n}}&0_n
\end{array}}&{}\\
{}&1
\end{array}} \right)$ and $S_{2n+2}=\left( {\begin{array}{*{20}{c}}
{\begin{array}{*{20}{c}}
{\begin{array}{*{20}{c}}
0_n&{{I_n}}\\
{{I_n}}&0_n
\end{array}}&{}\\
{}&1
\end{array}}&{}\\
{}&{ - z}
\end{array}} \right)$,  we construct matrices\\
$$A_2 = {\left( {\begin{array}{*{20}{c}}
{\begin{array}{*{20}{c}}
{{A_g}}&{}\\
{}&{{}^tA_g^{ - 1}}
\end{array}}&{}\\
{}&1
\end{array}} \right)_{ 2n + 1   }} {\rm{and}}\ A_3 = {\left( {\begin{array}{*{20}{c}}
{\begin{array}{*{20}{c}}
{\begin{array}{*{20}{c}}
{{A_g}}&{}\\
{}&{{}^tA_g^{ - 1}}
\end{array}}&{}\\
{}&1
\end{array}}&{}\\
{}&1
\end{array}} \right)_{ 2n + 2  }}\ , $$
where ${\rm{ord}}(A_2)={\rm{ord}}(A_3)=q^n-1$. Two linear rank metric codes can also be obtained with the parameters $(n\times {2n+1}, q^{2n}, n)_q$ and $(n\times {2n+2}, q^{2n}, n)_q$.

\subsection{Constructing MRD codes based on Abelian non-cyclic orthogonal groups}

Given two primitive polynomials of degree $n$ over $\mathbb{F}_q$:
\[f(x)= x^n +a_{n-1}x^{n-1}+ \cdots +a_1x+a_0,\]
\[g(x)= x^n +b_{n-1}x^{n-1}+ \cdots +b_1x+b_0.\]
The companion matrices of $f(x)$ and $g(x)$ are 
\[A_g= \left( {\begin{array}{*{20}{c}}
0&0& \cdots &0&{ - {a_0}}\\
1&0& \cdots &0&{ - {a_1}}\\
0&1& \cdots &0&{ - {a_2}}\\
 \vdots & \vdots &{}& \vdots & \vdots \\
0&0& \cdots &1&{ - {a_{n - 1}}}
\end{array}} \right)_{n },
B_g= \left( {\begin{array}{*{20}{c}}
0&0& \cdots &0&{ - {b_0}}\\
1&0& \cdots &0&{ - {b_1}}\\
0&1& \cdots &0&{ - {b_2}}\\
 \vdots & \vdots &{}& \vdots & \vdots \\
0&0& \cdots &1&{ - {b_{n - 1}}}
\end{array}} \right)_{n}.\]

For two matrices
\[A_1 = \left( {\begin{array}{*{20}{c}}
{{A_g}}&{}\\
{}&{{}^tA_g^{ - 1}}
\end{array}} \right),
B_1 = \left( {\begin{array}{*{20}{c}}
{{B_g}}&{}\\
{}&{{}^tB_g^{ - 1}}
\end{array}} \right),\]
where ${\rm{ord}}(A_1)={\rm{ord}}(B_1)=q^n-1$.
Let 
\[\langle A_1 \rangle =\{ {A^i_1} | 0 \le i \le q^n-2 \} \cong Z_{q^n - 1}\]
and
\[\langle B_1 \rangle =\{ {B^j_1} | 0 \le j \le q^n-2 \} \cong Z_{q^n - 1}.\]
Based on ${S_{4n}} = \left( {\begin{array}{*{20}{c}}
{{0_{2n}}}&{{I_{2n}}}\\
{{I_{2n}}}&{{0_{2n}}}
\end{array}} \right)$, 
we construct 
\[G_2 = \left\{ \left ({\begin{array}{*{20}{c}}
{{{A^i_1}}}&{}\\
{}&{{{B^j_1}}}
\end{array}}\right)  |    0 \le i,j \le q^n-2   \right  \}.\]
Clearly, $G_2$ is an Abelian non-cyclic orthogonal group over $\mathbb{F}_q$ with respect to $S_{4n}$.
Then $G_2 \cong \langle A_1 \rangle \times \langle B_1 \rangle $ and $G_2$ is an Abelian non-cyclic orthogonal group of order ${(q^n - 1)^2} $ over $\mathbb{F}_q$ with respect to $S_{4n}$.

Selecting the generating matrices of subspace
$$V_{1}=(I_n\quad 0_n\quad 0_n\quad 0_n)_{n \times {4n}},$$ 
$$V_{11}=(0_n\quad 0_n\quad I_n\quad 0_n)_{n \times {4n}}.$$
We construct codes 
$$\mathcal{D}_{1} =  V_{1}G_2 \cup  \{ {{0_{n \times 4n}}} \} ,$$
and
$$\mathcal{D}_{11} =  V_{11}G_2 \cup  \{ {{0_{n \times 4n}}} \}.$$

 \begin{lemma} \label{Le5}
Let $\mathcal{D}^{(1)} = { \mathcal{D}_1}  + {\mathcal{D}_{11}} $, then code $\mathcal{D}^{(1)}$ is a rank metric code with parameters $(n \times {4n}, q^{2n}, n )_q$.
\end{lemma}

\begin{proof}
Since
\[ \mathcal{D}_{1} =  V_{1}G_2 \cup { 0_{n \times {4n} }  }  =\{({A^i_g}\quad 0_n\quad 0_n\quad 0_n)   | 0 \le i \le q^n-2   \}  \cup \{ 0_{n \times {4n}}\} , \]
\[ \mathcal{D}_{11} =  V_{11}G_2 \cup { 0_{n \times {4n} }  }  =\{(0_n\quad 0_n\quad {B^j_g}\quad 0_n)   | 0 \le j \le q^n-2   \}  \cup \{ 0_{n \times {4n}}\} . \]
Then $ |\mathcal{D}_1|=| \mathcal{D}_{11}|=q^{n}$ and $d(\mathcal{D}_{1})=d(\mathcal{D}_{11})=n$. Let $\mathcal{D}^{(1)} = { \mathcal{D}_1}  + {\mathcal{D}_{11}}$, we obtain 
$\mathcal{D}^{(1)} = {\mathcal{D}_{1}} \oplus {\mathcal{D}_{11}}$,  $| \mathcal{D}^{(1)} | = |{\mathcal{D}_{1}} | | {\mathcal{D}_{11}} | =q^{2n}$ and ${d_R}( \mathcal{D}^{(1)}) = n$, 
which means that the code $\mathcal{D}^{(1)}$ is a rank metric code with parameters $(n \times {4n}, q^{2n}, n )_q$.
\end{proof}

Similarly, we select the generating matrices of subspace
$$ {V^{\prime}_{1}}=(0_n\quad I_n\quad 0_n\quad 0_n)_{n \times {4n}}, {V^{\prime}_{11}}=(0_n\quad 0_n\quad 0_n\quad I_n)_{n \times {4n}}.$$ 
We construct
\[ {\mathcal{D}^{\prime}_{1}} = {V^{\prime}_{1}} G_2= \{(0_n\quad ({{}^tA_g^{ - 1}})^i\quad 0_n\quad 0_n)   | 0 \le i < q^n-1   \}  \cup \{ 0_{n \times {4n}}\} , \]
\[ {\mathcal{D}^{\prime}_{11}} = {V^{\prime}_{11}} G_2=  \{(0_n\quad 0_n\quad 0_n\quad ({{}^tB_g^{ - 1}})^j)   | 0 \le j < q^n-1   \}  \cup \{ 0_{n \times {4n}}\} . \]
Let $ {\mathcal{D}^{ (2)}} = {\mathcal{D}^{\prime}_1}  +  {\mathcal{D}^{\prime}_{11}}  $. Similarly, code ${\mathcal{D}^{ (2)}}$ is a rank metric code with the parameters $(n \times {4n}, q^{2n}, n )_q$.

\begin{theorem} \label{Th6}
Let $${\cal C}^{(2)} ={\mathcal{D}^{ (1)}} +  {\mathcal{D}^{ (2)}} = \{d_1+{d_2}  | {d_1}\in {\mathcal{D}^{ (1)}}  , {d_2} \in {{\mathcal{D}^{ (2)}} }\},$$ then ${\cal C}^{(2)}$ is an MRD code with parameters $(n \times {4n}, q^{4n} , n)_q$.
\end{theorem}

\begin{proof} For ${\mathcal{D}^{ (1)}}  \cap { {\mathcal{D}^{ (2)}}}= \{ {{0_{n \times 4n}}} \}$, $|{\cal C}^{(2)} |=|{\mathcal{D}^{ (1)}}|| { {\mathcal{D}^{ (2)}}}|=q^{4n}$. Furthermore, $d_R({\cal C}^{(2)} )= \min \{d_R{\mathcal{D}^{ (1)}}), d_R({ {\mathcal{D}^{ (2)}}})\}=n$. In summary, ${\cal C}^{(2)}$ is a MRD code with parameters $(n \times {4n}, q^{4n} , n)_q$.
\end{proof}

Similarly, Based on $$S_{4n+1, 1}=\left( {\begin{array}{*{20}{c}}
{0_n}&{{I_n}}&{}&{}&{}\\
{{I_n}}&{0_n}&{}&{}&{}\\
{}&{}&{0_n}&{{I_n}}&{}\\
{}&{}&{{I_n}}&{0_n}&{}\\
{}&{}&{}&{}&{1}
\end{array}} \right), S_{4n+2}=\left( {\begin{array}{*{20}{c}}
{0_n}&{{I_n}}&{}&{}&{}&{}\\
{{I_n}}&{0_n}&{}&{}&{}&{}\\
{}&{}&{0_n}&{{I_n}}&{}&{}\\
{}&{}&{{I_n}}&{0_n}&{}&{}\\
{}&{}&{}&{}&{1}&{}\\
{}&{}&{}&{}&{}&{-z}
\end{array}} \right),$$
we can also obtain rank metric codes with parameters $(n \times {4n+1}, q^{4n} , n)_q$ and $(n \times {4n+2}, q^{4n} , n)_q$.


\subsection{Constructing sum-rank metric codes based on MRD codes}
In this section, we concentrate on the construction of the sum-rank metric codes. Next, we used two methods to construct sum-rank metric codes with different parameters based on the above linear MRD codes.

\noindent {\bf Construction 1}
Taking an $(2 \times 4, q^{4}, 2 )_q$-MRD code ${\cal{C}}^{(1)}$ in Theorem \ref{Th4}.
Consider the map
\[\begin{array}{*{20}{c}}
{{\phi _i}:}&{\mathbb{F}_q^{{m_i} \times 4}}& \to &{\mathbb{F}_q^{2 \times 4}}\\
{}&{{C_i}}& \mapsto &{^t({C_i}\quad {0_{2 - {m_i}}})}
\end{array}\]

The positive integers $m_i$ satisfy the condition $m_1=2 \ge m_2 \ge \cdots \ge m_t$, then ${\phi _i}$ is rank-preserving and linear.

\begin{theorem} \label{Th7} 
Define the sum-rank metric code  
\[{\overline {\cal C} _1} = \left\{ {({C^{(1)}} - \sum\limits_{i = 2}^t {{\phi _i}({C_i})} |{C_2}| \cdots |{C_t})\left| {{C_i} \in {\mathbb{F}}_q^{{m_i} \times 4},{C^{(1)}} \in {{\cal C}^{(1)}}} \right.} \right\},\]
then $ \overline{\mathcal{C}}_1$ is an MSRD code with a distance of $2$.
\end{theorem}
\begin{proof} 
Clearly, $\dim_q ( \overline{\mathcal{C}}_1) =4 \sum\limits_{i = 2}^t {m_i} +4$, it can be seen that $w_{SR}(C^{(1)}-\sum\limits_{i = 2}^t {{\phi _i}( C_i )} | C_2| \cdots| C_t) \ge 2$ for all nonzero elements in $ \overline{\mathcal{C}}_1$. This is clear if either $w_{SR}( C_2| \cdots| C_t) \ge 2$ or if $w_{SR}( C_2| \cdots| C_t) =0$. Thus let $w_{SR}( C_2| \cdots| C_t) =1$, ${\rm{rank}}(C_l)=1$ for some $l$ and $C_i={\bf 0}$ for $i \ne  l$. Thus
$$C^{(1)}-\sum\limits_{i = 2}^t {{\phi _i}( C_i )} =C^{(1)}-  {{\phi _l}( C_l )}.$$
Since $C^{(1)} \in  \mathcal{C}^{(1)} $ and $C^{(1)} \ne {\bf 0}$, we get ${\rm{rank}}(C^{(1)})\ne 1$, and thus
$${\rm{rank}}(C^{(1)}-\sum\limits_{i = 2}^t {{\phi _i}( C_i )})\ne 0.$$
The assumption is contradictory, 
$$w_{SR}(C^{(1)}-\sum\limits_{i = 2}^t {{\phi _i}( C_i )}|  C_2| \cdots| C_t) \ge 2,$$ 
i.e., $d(\overline{\mathcal{C}}_1) \ge 2.$ In addition, it follows from the Singleton upper bound that
$$ |{\overline{\mathcal{C}}_1}| \le q^{n(\sum\limits_{i = 1}^t m_i-d+1)},$$
we get $d(\overline{\mathcal{C}}_1) \le 2$. All of this shows that $\overline{\mathcal{C}}_1$ is an MSRD code with distance 2, the parameters of the MSRD code $\overline{\mathcal{C}}_1$ are $((2\times4| m_2\times4 | \cdots | m_t\times4), q^{4(\sum\limits_{i = 2}^t {m_i} +1)}, 2)_q$.
\end{proof}

In particular, when $m_2=\cdots =m_t=2$, the parameters of code $\overline{\mathcal{C}}_1$ are $((2\times4 | \cdots | 2\times4), q^{4(2t-1)}, 2)_q$. According to Definition \ref{Df5}, the parameters of vector sum-rank metric code $\overline {\mathcal{C}}_1$ are $[(4| \cdots | 4), 2(2t-1), 2]_{{q^2}/q}$.

Taking $(2 \times 8, q^{8}, 2 )_q$-MRD code ${\cal{C}}^{(2)}$ in Theorem \ref{Th6}, we can also onstruct MSRD code.

\begin{corollary}\label{Co8} 
Define the sum-rank metric code 
\[ \overline{\mathcal{C}}_2=\left\{(C^{(2)}-\sum\limits_{i = 2}^t {{\phi _i}( D_i )}   | D_2| \cdots| D_t) | D_i \in  {\mathbb{F}}_q^{ m_i \times 8} , C^{(2)} \in   \mathcal{C}^{(2)}    \right\}   , \]
where $m_1=2 \ge m_2 \ge \cdots \ge m_t$, then the parameters of MSRD code $\overline{\mathcal{C}}_2$ are $((2\times8 | m_2\times8 | \cdots | m_t\times8), q^{8(\sum\limits_{i = 2}^t {m_i} +1)}, 2)_q$.
\end{corollary}

\noindent {\bf Construction 2}
Using an $(n \times 2n, q^{2n}, n )_q$-MRD code in Theorem \ref{Th4}, we denote it by ${\cal{C}}_t$. Let $C_{t,1},C_{t,2}, \cdots,C_{t,2n}$ be a $\mathbb{F}_q$-basis of ${\mathcal{C}_t}$. Choose $(n \times n, q^{2n}, n-1 )_q$-MRD codes ${\cal{A}}_i$ for $i = 1, 2, \cdots, t-1$. Let $A_{i,1},A_{i,2},\cdots,A_{i,2n} $ be a $\mathbb{F}_q$-basis of ${\mathcal{A}_i}$.

\begin{theorem} \label{Th9} 
Define the sum-rank metric code  
\[\overline{\mathcal{C}}_3 = \left\langle {\left( {{A_{1,1}}| \cdots|{A_{t-1,1}}|{C_{t,1}}} \right) \cdots \left( {{A_{1,{2n}}}| \cdots|{A_{t-1,{2n}}}|{C_{t,{2n}}}} \right)} \right\rangle, \]
then the parameters of sum-rank metric code $ \overline{\mathcal{C}}_3$ are $[(n | \cdots|n | 2n), 2, t(n-1)+1]_{{q^n}/q}$.
\end{theorem}

\begin{proof} 
Clearly, $\dim_q(\overline{\mathcal{C}}_3 )=2n$. 
\begin{align}
d(\overline{\mathcal{C}}_3)&=\min \{w_{SR}(A_{1,1}| \cdots |{A_{t-1,1}}|C_{t,1}), \cdots, w_{SR}(A_{1,2n}| \cdots|{A_{t-1,2n}} |C_{t,2n})\} \nonumber \\
&=\min \{ (\sum\limits_{i=1}^{t-1} {\rm{rank}}(A_{i, 1}))+{\rm{rank}}(C_{t, 1}), \cdots, (\sum\limits_{i=1}^{t-1} {\rm{rank}}(A_{i, 2n}))+{\rm{rank}}(C_{t, 1}) \} \nonumber   \\
&=  t(n-1)+1.\nonumber
\end{align}
Thus the parameters of $\overline{\mathcal{C}}_3$ are $((n\times n | \cdots |n\times n | n\times2n), q^{2n}, t(n-1)+1)_q$. According to Definition \ref{Df5}, the parameters of $\overline{\mathcal{C}}_3$ are $[(n | \cdots|n | 2n), 2, t(n-1)+1]_{{q^n}/q}$.
\end{proof}

Similarly, choose MRD codes with parameters $(n\times(2n+1), q^{2n+1}, n)_q$ and $(n\times(2n+2), q^{2n+2}, n)_q$, we can also construct sum-rank metric codes with parameters $((n\times n | \cdots |n\times n | n\times(2n+1)), q^{2n+1}, t(n-1)+1)_q$ and $((n\times n | \cdots |n\times n | n\times(2n+2)), q^{2n+2}, t(n-1)+1)_q$.

Taking an $(n \times 4n, q^{4n}, n )_q$-MRD code in Theorem \ref{Th6}, we denote it by ${\cal{D}}_t$. Let $D_{t,1},D_{t,2},\cdots,D_{t,4n} $ be a $\mathbb{F}_q$-basis of ${\mathcal{D}_t}$. Choose $(n \times n, q^{4n}, n-3 )_q$-MRD codes ${\cal{B}}_i$ for $i = 1, 2, \cdots, t-1$. Let $B_{i,1},B_{i,2},\cdots,B_{i,4n} $ be a $\mathbb{F}_q$-basis of ${\mathcal{B}_i}$.

\begin{corollary} \label{Co10} 
Define the sum-rank metric code  
\[\overline{\mathcal{C}}_4 = \left\langle {\left( {{B_{1,1}}| \cdots|{B_{t-1,1}} |{D_{t,1}}} \right), \cdots ,\left( {{B_{1,{4n}}}| \cdots |{B_{t-1,{4n}}}|{D_{t,{4n}}}} \right)} \right\rangle, \]
then the parameters of the sum-rank metric code $ \overline{\mathcal{C}}_4$ are $[(n | \cdots|n | 4n), 4, t(n-3)+3]_{{q^{n}}/q}$.
\end{corollary}

For MRD codes with parameters $(n\times(4n+1), q^{4n+1}, n)_q$ and $(n\times(4n+2), q^{4n+2}, n)_q$, we can also  construct sum-rank metric codes with parameters $((n\times n | \cdots|n\times n  |n\times(4n+1)), q^{4n+1}, t(n-1)+1)_q$ and $((n\times n | \cdots|n\times n  | n\times(4n+2)), q^{4n+2}, t(n-1)+1)_q$.

For convenience, we used $\overline{\mathcal{C}}$ to denote the sum-rank metric codes $\overline{\mathcal{C}_i}$ for $i =1,2, 3, 4$.

\section{List decoding of sum-rank metric codes}\label{sec4}

\subsection{List size of sum-rank metric codes under list decoding algorithm}

In this subsection, we use list decoding algorithm to compute the list size of sum-rank metric code $\overline{\mathcal{C}}$.

\begin{lemma}\label{Le11}  If non-negative integers $n$ and $r$ satisfy $r\le n$, then the number of $r$-dimensional subspaces over ${\mathbb{F}}^n_q$ is denoted by ${\genfrac[]{0pt}{0}{n}{r}_q }$, which satisfies
$${\genfrac[]{0pt}{0}{n}{r}_q }< {\gamma _q}{q^{r\left( {n - r} \right)}},$$
where ${\gamma _q} = \prod\limits_{i = 1}^\infty  {{{\left( {1 - {q^{ - i}}} \right)}^{ - 1}}} $ monotonically decreasing and infinitely close to $1$.  
\end{lemma}

\begin{proof}
Since
\begin{equation}
{\genfrac[]{0pt}{0}{n}{r}_q } = \prod\limits_{i = 1}^r {\frac{{{q^{n - r + i}} - 1}}{{{q^i} - 1}}}. \nonumber 
\end{equation}
When $r=0$, ${\genfrac[]{0pt}{0}{n}{0}_q }= 1$. We get 
\begin{align}
{\genfrac[]{0pt}{0}{n}{r}_q }
&= {q^{r\left( {n - r} \right)}}\frac{{\left( {1 - {q^{ - n}}} \right)\left( {1 - {q^{ - n + 1}}} \right) \cdots \left( {1 - {q^{ - n + r - 1}}} \right)}}{{\left( {1 - {q^{ - r}}} \right)\left( {1 - {q^{ - r + 1}}} \right) \cdots \left( {1 - {q^{ - 1}}} \right)}}\nonumber \\
&< {q^{r\left( {n - r} \right)}}\frac{1}{{\left( {1 - {q^{ - r}}} \right)\left( {1 - {q^{ - r + 1}}} \right) \cdots \left( {1 - {q^{ - 1}}} \right)}}\nonumber \\
&< {q^{r\left( {n - r} \right)}}\prod\limits_{i = 1}^\infty  {{{\left( {1 - {q^{ - i}}} \right)}^{ - 1}}}.\nonumber
\end{align}
Let ${\gamma _q}= \prod\limits_{i = 1}^\infty  {{{\left( {1 - {q^{ - i}}} \right)}^{ - 1}}}$, which is monotonically decreasing and infinitely close to $1$. Thus
$${\genfrac[]{0pt}{0}{n}{r}_q }< {\gamma _q}{q^{r\left( {n - r} \right)}}.$$
\end{proof}

\begin{definition}(See {\cite{ref207}}) \label{Df8} Define the set \[{\omega} = \left\{ { \left( {{w_1},{w_2}, \cdots ,{w_t}} \right)|\sum\limits_{i = 1}^t {{w_i} = w,} {w_i} \le \mu } \right\}.\] The set is $t$-decomposition of the integer $w$ (each part does not exceed $\mu$).
\end{definition} 
Because this set ${\omega}$ is composed of the coefficients of $x^j$ in the polynomial ${( {\sum\limits_{j = 0}^\mu  {{x^j}} } )^t}$, we can derive 
\begin{equation}
| \omega | = \sum\limits_{i = 0}^{\left\lfloor {\frac{w}{\mu  + 1}} \right\rfloor } {{( - 1)^i}{\genfrac[]{0pt}{0}{t}{i}_q } {\genfrac[]{0pt}{0}{w+t-1-(\mu+1)i}{t-1}_q }    }  \le {\genfrac[]{0pt}{0}{w+t-1}{t-1}_q }.\nonumber
\end{equation}
The detailed proof can be found in the reference {\cite{ref207}}.

\begin{lemma}\label{Le12}  
Let $n(n, n_i, r)$ be the number of $n \times n_i$ matrices ($r \le n$) of rank $r$ over $\mathbb{F}_q$, which satisfies
\begin{equation}
n(n, n_i, r) < q^{r( {n + n_i - r})}.  \nonumber
\end{equation}
\end{lemma}

\begin{proof}
For $n(n, n_i, r)$, we have
\begin{align}
n(n, n_i, r) &= {\genfrac[]{0pt}{0}{n}{r}_q }\left( {{q^{n_i}} - 1} \right)\left( {{q^{n_i}} - q} \right) \cdots \left( {{q^{n_i}} - {q^{r - 1}}} \right)\nonumber \\
&= {\genfrac[]{0pt}{0}{n_i}{r}_q }\left( {{q^n} - 1} \right)\left( {{q^n} - q} \right) \cdots \left( {{q^n} - {q^{r - 1}}} \right)\nonumber \\
&= \prod\limits_{j = 1}^{r } q^{\frac{r(r-1)}{2}} {\frac{{\left( {{q^{n-r+j}} - 1} \right)\left( {{q^{n-j+1}} - 1} \right)}}{{{q^i} - {1}}}}.\nonumber
\end{align}
Thus
\begin{align}
n\left( {n,n_i,r} \right) &= {\genfrac[]{0pt}{0}{n_i}{r}_q }\prod\limits_{j = 0}^{r - 1} {\left( {{q^n} - {q^j}} \right)}\nonumber \\
&< {\gamma _q}{q^{r\left( {n_i - r} \right)}}{q^{nr}}\prod\limits_{j = n - r + 1}^n {\left( {1 - {q^{ - j}}} \right)} \nonumber\\
&< {\gamma _q}{q^{r\left( {n + n_i - r} \right)}}\prod\limits_{j = 1}^\infty  {\left( {1 - {q^{ - j}}} \right)} \nonumber\\ 
&= {q^{r\left( {n + n_i - r} \right)}}. \nonumber
\end{align}
\end{proof}

We construct a sphere with any element $x$ in $\mathbb{F}^{\bf n}_{q^n}$ as the center and covering radius $\tau_{SR}(\overline{\mathcal{C}})$ as the radius. Thus
$$\overline{\mathcal{C}} \cap \mathcal{B}_{SR}(x, \tau_{SR}(\overline{\mathcal{C}}))=  \sum\limits_{r = 0}^{\tau_{SR}(\overline{\mathcal{C}})}  {( {\overline{\mathcal{C}} \cap \mathcal{S}( x,r)} )} ,$$
where $\mathcal{S}( x, r )$ represents all elements on the surface of the sphere in $\mathbb{F}^{\bf n}_{q^n}$ with center $x$ and radius $r$. 

\begin{lemma}\label{Le13}
When $t >1$, 
\[\left| {\mathcal{S}\left( {x,r} \right) \cap \overline{\mathcal{C}}} \right| \le {\genfrac[]{0pt}{0}{r+t-1}{t-1}_q }{q^{r\left( {n + \max\limits_{i=\{1,2,\cdots, t\}}\{n_i\}} \right) - \frac{{{r^2}}}{t}}}.\]
\end{lemma}

\begin{proof} By Definition \ref{Df8}, we perform a $t$-decomposition of integer $r$ and obtain the set 
$$\omega=\{(r_1,r_2, \cdots, r_t)| \sum\limits_{j=1}^t{r_j}=r, 0 \le r_j \le r\}.$$
When $t>1$, 
\begin{align}
\left| {\mathcal{S}\left( {x,r} \right) \cap \overline{\mathcal{C}}} \right| &\le \sum\limits_{(r_1,r_2, \cdots, r_t) \in {\omega}} ({\prod\limits_{j = 1}^t {n\left( {n,n_i,{r_j}} \right)} }) \nonumber\\
&\le \left| {{\omega }} \right|\mathop {\max }\limits_{(r_1,r_2, \cdots, r_t) \in {\omega}} \left\{ {\prod\limits_{j = 1}^t {n\left( {n,n_i,{r_j}} \right)} } \right\}.
\nonumber
\end{align}
By Lemma \ref{Le12},
\[| {\mathcal{S}(x,r) \cap \overline{\mathcal{C}}} | \le {\genfrac[]{0pt}{0}{r+t-1}{t-1}_q }{q^{\mathop {\max }\limits_{(r_1,r_2, \cdots, r_t) \in {\omega }} \left\{ {\sum\limits_{j = 1}^t {{r_j}\left( {n + n_i - {r_j}} \right)} } \right\}}}.\]
Taking the upper bound when $r_j = \frac{r}{t}$, then \[\mathop {\max }\limits_{(r_1,r_2, \cdots, r_t) \in {\omega }} \left\{ {\sum\limits_{j = 1}^t {{r_j}\left( {n + n_i - {r_j}} \right)} } \right\} \le r\left( {n + \max\limits_{i=\{1,2,\cdots, t\}}\{n_i\}} \right) - \frac{{{r^2}}}{t}.\]
Thus we get 
\[\left| {\mathcal{S}\left( {x,r} \right) \cap \overline{\mathcal{C}}} \right| \le {\genfrac[]{0pt}{0}{r+t-1}{t-1}_q }{q^{r\left( {n + \max\limits_{i=\{1,2,\cdots, t\}}\{n_i\}} \right) - \frac{{{r^2}}}{t}}}.\]
\end{proof}

\begin{theorem}\label{Th14} 
Let $t >1$, the list size of $[{\bf n}, k, d]_{{q^n}/q}$-sum-rank metric codes $\overline{\mathcal{C}}$ under the list decoding algorithm is
$$L < 1 + \left( {{\tau _{SR}}(\overline {\cal C} ) - \left\lfloor {\frac{{d - 1}}{2}} \right\rfloor } \right){q^{{\tau _{SR}}(\overline {\cal C} ){\rm{ }}\left( {n + \mathop {\max }\limits_{i = \{ 1,2, \cdots ,t\} } \{ {n_i}\}  + t - 1} \right)}}.$$
\end{theorem}

\begin{proof}
According to Definition \ref{Df6},
$$L \le  \left| {\mathcal{B}_{SR}\left( {x,\tau_{SR}(\overline{\mathcal{C}}) } \right) \cap \overline{\mathcal{C}}} \right| = \sum\limits_{r = 0}^{\tau_{SR}(\overline{\mathcal{C}})}  {\left| {\mathcal{S}\left( {x,r} \right) \cap \overline{\mathcal{C}}} \right|}.$$ 
For a sphere with $x$ as its center and ${\left\lfloor {\frac{{d - 1}}{2}} \right\rfloor }$ as its radius, there can be at most one codeword in $\overline{\mathcal{C}}$. Thus 
$$\sum\limits_{r = 0}^{\left\lfloor {\frac{{d - 1}}{2}} \right\rfloor } {\left| {\mathcal{S}\left( {x,r} \right) \cap \overline{\mathcal{C}}} \right|} \le 1.$$ We get
$$L \le 1 + \sum\limits_{r = \left\lfloor {\frac{{d - 1}}{2}} \right\rfloor  + 1}^{\tau_{SR}(\overline{\mathcal{C}})}  {\left| {\mathcal{S}\left( {x,r} \right) \cap \overline{\mathcal{C}}} \right|} .$$
Next we compute $ {\left| {\mathcal{S}\left( {x,r} \right) \cap \overline{\mathcal{C}}} \right|}$ in order to find the list size of sum-rank metric codes $\overline{\mathcal{C}}$.
By Lemma \ref{Le13},
\[\left| {\mathcal{S}\left( {x,r} \right) \cap \overline{\mathcal{C}}} \right| \le {\genfrac[]{0pt}{0}{r+t-1}{t-1}_q }{q^{r\left( {n + \max\limits_{i=\{1,2,\cdots, t\}}\{n_i\}} \right) - \frac{{{r^2}}}{t}}}.\]
Thus the list size of sum-rank metric codes $\overline{\mathcal{C}}$ are
\begin{align}
L & \le 1 + \sum\limits_{r = \left\lfloor {\frac{{d - 1}}{2}} \right\rfloor  + 1}^{\tau_{SR}(\overline{\mathcal{C}})}  {\left| {\mathcal{S}\left( {x,r} \right) \cap \overline{\mathcal{C}}} \right|} \nonumber\\
& \le 1 + \sum\limits_{r = \left\lfloor {\frac{{d - 1}}{2}} \right\rfloor  + 1}^{\tau_{SR}(\overline{\mathcal{C}})}  {\genfrac[]{0pt}{0}{r+t-1}{t-1}_q } {q^{r\left( {n + \max\limits_{i=\{1,2,\cdots, t\}}\{n_i\} - \frac{r}{t}} \right)}}\nonumber\\
&< 1 + \sum\limits_{r = \left\lfloor {\frac{{d - 1}}{2}} \right\rfloor  + 1}^{\tau_{SR}(\overline{\mathcal{C}})}  {{q^{\left( {t - 1} \right)r}}} {q^{i\left( {n + \mathop {\max }\limits_{i= \{ 1,2, \cdots ,t\} } \{ {n_i}\}  - \frac{r}{t}} \right)}}\nonumber\\
& \le 1 + \sum\limits_{r = \left\lfloor {\frac{{d - 1}}{2}} \right\rfloor  + 1}^{\tau_{SR}(\overline{\mathcal{C}})}  {{q^{\left( {t - 1} \right)r}}} {q^{r\left( {n + \mathop {\max }\limits_{i= \{ 1,2, \cdots ,t\} } \{ {n_i}\}  - \frac{{\left\lfloor {\frac{{d - 1}}{2}} \right\rfloor  + 1}}{t}} \right)}}\nonumber\\
&< 1 + \left( {{\tau _{SR}}(\overline {\cal C} ) - \left\lfloor {\frac{{d - 1}}{2}} \right\rfloor } \right){q^{{\tau _{SR}}(\overline {\cal C} ){\rm{ }}\left( {n + \mathop {\max }\limits_{i = \{ 1,2, \cdots ,t\} } \{ {n_i}\}  + t - 1} \right)}}. \nonumber
\end{align}
\end{proof}

Let $\lfloor {\frac{d-1}{2}} \rfloor<t\le N-k$ and $\tau_{SR}(\overline{\mathcal{C}})=t$, we can obtain the list size of sum-rank metric codes $\overline{\mathcal{C}}$ under the list decoding algorithm.

\begin{corollary}\label{Co15} 
The list size of MSRD code $\overline{\mathcal{C}}_1$ with the parameters $[(4|\cdots|4),2(2t-1),2]_{{q^2}/q}$ is at most $1+tq^{t(5+t)}$.  
\end{corollary} 

\begin{corollary}\label{Co16} 
The list size of MSRD code $\overline{\mathcal{C}}_2$ with the parameters $[(8|\cdots|8),4(2t-1),2]_{{q^2}/q}$ is at most $1+tq^{t(5+t)}$. 
\end{corollary}

\begin{corollary}\label{Co17}
The list size of sum-rank metric code $\overline{\mathcal{C}}_3$ with the parameters $[(n | \cdots|n | 2n), 2, t(n-1)+1]_{{q^n}/q}$ is at most $1+(t-\lfloor {\frac{t(n-1)}{2}} \rfloor)q^{t(3n+t-1)}$. 
\end{corollary} 

\begin{corollary}\label{Co18} 
The list size of sum-rank metric code $\overline{\mathcal{C}}_4$ with the parameters $[(n|\cdots|n|4n),4,t(n-3)+3]_{{q^n}/q}$ is at most $1+(t-\lfloor {\frac{t(n-3)}{2}+1} \rfloor)q^{t(5n+t-1)}$. 
\end{corollary}

\subsection{The relationship between sum-rank metric codes and subspace designs}


Subspace design is a coding method based on subspace division, that improves the decoding success rate by finding low-dimensional subspaces to reduce the list size . In this section, we first study the relationship between subspace designs and sum-rank metric codes, and then calculate the list size of sum-rank metric codes in Section \ref{sec3} based on subspace designs. Finally, we compare the results to conclude that the coding method based on subspace designs can improve the decoding success rate.

\begin{definition}  ($[{\bf{n}},k,d]_{{q^m}/q}$-system, see \cite{ref204})  \label{Df10}
Let ${\bf{n}}=(n_1, \cdots,n_t) \in \mathbb{N}^t$, with $n_i \ge \cdots \ge n_t$. An $[{\bf{n}},k,d]_{{q^m}/q}$-system $\mathcal{H}= (H_1, H_2, \cdots, H_{t})$ is an ordered set. For any $i = 1,2, \cdots ,t$, $H_i$ is an $\mathbb{F}_q$-subspace of $\mathbb{F}_{{q^m}}^{{k}}$ with dimension $n_i$, which satisfies the following conditions: 
$$  \langle  H_1, H_2, \cdots, H_{t} \rangle  _{q^m}   =        \mathbb{F}_{{q^m}}^{{k}}         $$
and
$$N - d = {\rm{max}}   \{     \sum\limits_{i = 1}^{t} {{{\dim }_q}( {{H_i} \cap P} )}       \}  ,    $$
where $ P \subseteq \mathbb{F}_{{q^m}}^{{k}}$ and $\dim_{q^m}( P)=k-1$.
\end{definition}

For a sum-rank metric code $\overline{\mathcal{C}}$ with the parameters $[{\bf n}, k, d]_{{q^n}/q}$. Let $\overline{G}= {\left( {{G_{1}}|{G_{2}}| \cdots |{G_{{t}}}} \right)_{{k} \times N}}$ be a generator matrix for $\overline{\mathcal{C}}$, with $N=n_1 +n_2 + \cdots +n_{t}$. We denote the set of equivalence classes of ${[{\bf{n}},k,d]_{{q^n}/q}}$ sum-rank metric codes as $\overline{\mathcal{C}}{[{\bf{n}},k,d]_{{q^n}/q}}$. Moreover, the set of equivalence classes of $[{\bf{n}},k,d]_{{q^n}/q}$-systems by $\overline{\mathcal{H}}[{\bf{n}},k,d]_{{q^n}/q}$. 

Define two maps
$$ \Phi :{\overline{\mathcal{C}}}{[{\bf{n}},k,d]_{{q^n}/q}} \to {\overline{\mathcal{H}}}{[{\bf{n}},k,d]_{{q^n}/q}}    ,$$
$$  \Psi :{\overline{\mathcal{H}}}{[{\bf{n}},k,d]_{{q^n}/q}} \to {\overline{\mathcal{C}}}{[{\bf{n}},k,d]_{{q^n}/q}}   .$$

Let $\overline{\mathcal{C}} \in {\overline{\mathcal{C}}}{[{\bf{n}},k,d]_{{q^n}/q}}$, then $\Phi(\overline{\mathcal{C}})$ is the equivalence class of $[{\bf{n}},k,d]_{{q^n}/q}$-system $\mathcal{H}=(H_{1},H_{2},\cdots,\\ H_{t})$. $H_{j}$ is the $\mathbb{F}_q$-span of the columns of $G_{j}$ for every $j \in[t]$. Similarly, given $[{\bf{n}},k,d]_{{q^n}/q}$-system $\mathcal{H} \in {\overline{\mathcal{H}}}{[{\bf{n}},k,d]_{{q^n}/q}}$. Let $\{h_{1,j},h_{2,j}, \cdots, h_{n_j,j} \}$ be a $\mathbb{F}_q$-basis of $H_{j}$ and  
${G_{j}} =  {}^t (h_{1,j},h_{2,j}, \cdots, h_{n_j,j})$,
then $\Psi ({\cal H}) $ is the equivalence class of the sum-rank metric codes generated by $\bar G = {\left( {{{ G}_{1}} | {{ G}_{2}} | \cdots | {{ G}_{t}}} \right)}$.

Note that Neri et al.\textsuperscript{\cite{ref204}} have shows that maps $ \Phi $ and $ \Psi $ are bijective. Therefore, the equivalence classes $\overline{\mathcal{C}}{[{\bf{n}},k,d]_{{q^n}/q}}$ of the sum-rank metric codes $\overline{\mathcal{C}}$ and the equivalence classes $\overline{\mathcal{H}}[{\bf{n}},k,d]_{{q^n}/q}$ of $[{\bf{n}},k,d]_{{q^n}/q}$-systems are inversely proportional to each other.

\begin{theorem}\label{Th19} 
Let $\mathcal{H}= (H_{1},H_{2},\cdots, H_{t})$ be the set of ${\mathbb{F}}_q$-subspaces over $\mathbb{F}_{{q^n}}^{k}$, then $\mathcal{H}$ is the $[{\bf{n}},k,d]_{{q^n}/q }$-system if and only if $\mathcal{H}$ is an $(k-1,N-d) _q$-subspace design over $\mathbb{F}_{{q^n}}^{k}$. 
\end{theorem}

\begin{proof}
If $\mathcal{H}= (H_{1},H_{2},\cdots, H_{t})$ is an $[{\bf{n}},k,d]_{{q^n}/q}$-system over $\mathbb{F}_{{q^n}}^{k}$, there is $ \dim_{q^n} \langle  H_{1},H_{2},\cdots, H_{t}\\ \rangle  _{q^n}   \le  k-1 $ holds by Definition \ref{Df10}. For any subspace $P$ of dimension $k-1$ over $\mathbb{F}_{{q^n}}^{k}$, this satisfies
$$N - d = {\text{max}}   \{     \sum\limits_{i = 1}^{t} {{{\dim }_q}( {{H_i} \cap P} )}       \} .$$
Thus $$   \sum\limits_{i = 1}^{t} {{{\dim }_q}( {{H_i} \cap P} )}    \le  N - d .$$

According to Definition \ref{Df7}, $\mathcal{H}$ is an $(k-1,N-d) _q$-subspace design. Similarly, if $\mathcal{H}= (H_{1},H_{2},\cdots, H_{t})$ is an $(k-1,N-d) _q$-subspace design, it can be shown that $\mathcal{H}$ is an $[{\bf{n}},k,d]_{{q^n}/q}$-system.
\end{proof}

Based on the above theorem, the following conclusions were obtained from the sum-rank metric codes constructed in Section \ref{sec3}.

\begin{corollary}\label{Co20} 
The MSRD code $\overline{\mathcal{C}}_1$ with the parameters $[(4|\cdots|4),2(2t-1),2]_{{q^2}/q}$ corresponds to the $(4t-3, 4t-2)_q$-subspace design. 
\end{corollary} 

\begin{corollary}\label{Co21} 
The MSRD code $\overline{\mathcal{C}}_2$ with the parameters $[(8|\cdots|8),4(2t-1),2]_{{q^2}/q}$ corresponds to the $(8t-5, 8t-2)_q$-subspace design. 
\end{corollary}

\begin{corollary}\label{Co22}
The sum-rank metric code $\overline{\mathcal{C}}_3$ with the parameters $[(n | \cdots|n | 2n), 2, t(n-1)+1]_{{q^n}/q}$ corresponds to the $(1, n+t-1)_q$-subspace design. 
\end{corollary} 

\begin{corollary}\label{Co23} 
The sum-rank metric code $\overline{\mathcal{C}}_4$ with the parameters $[(n|\cdots|n|4n),4,t(n-3)+3]_{{q^n}/q}$ corresponds to the $(3, 3(n+t-1))_q$-subspace design. 
\end{corollary}

\subsection{List size of sum-rank metric codes based on subspace designs}



For a vector ${x}=(x_1, x_2, \cdots ,x_{tk}) \in \mathbb{F}^{kt}_{q^n}$ and positive integers $t_1 \le t_2 \le {kt} $, we denote by ${\rm{proj}}_{[t_1 ,t_2]}({ x}) \in \mathbb{F}_{q^n}^{t_2 - t_1 +1}$ its projection onto coordinates $t_1$ through $t_2$, i.e., 
$$  {\rm{proj}}_{[t_1 ,t_2]}({ x}) =(x_{t_1}, x_{t_1+1}, \cdots , x_{t_2}). $$

 \begin{definition}
(Periodic subspaces, see \cite{ref311}) \label{Df11}
For positive integers $s$, $n$, $k$ and $t$, a subspace $T \subseteq \mathbb{F}_{q^n}^{k} $ is called an $(s,k,t)$-periodic subspace if there exists a subspace $M \subseteq \mathbb{F}_{q^n}^k $ of dimension at most $s$ such that for $j \in \{1,2, \cdots ,t\}$, the set
\[ \{  {\rm{proj}}_{[(j-1)n+1 ,jn]}(x)\ | \ x \in T ,   {\rm{proj}}_{(j-1)n}(x) = v  \}\]
is contained in a coset of $M$.
\end{definition}

The above definition introduces the periodic subspace that emerges in list decoding algorithms. The introduction of such a subspace enables the computation of the list size of code $\overline{\mathcal{C}}$.



\begin{theorem} \label{Th25}
Let $T$ be an $(k-1,k,t)$-periodic subspace of $\mathbb{F}_{q^n}^{k}$, then the dimension of set 
$$S=\{(f_{1},f_{2},\cdots,f_{t})\in T| f_{j} \in H_{j} , j=1, 2, \cdots, t\}$$
has at most $N-d$, where ${\cal H} = (H_{1},H_{2},\cdots, H_{t})$ is an $(k-1, N-d)_q$-subspace design corresponding to the sum-rank metric code $\overline{\mathcal{C}}$.
\end{theorem}

\begin{proof} 
According to the set $S$, 
$$f_{1} \in  {\rm{proj}}_{[1,k]}(T) \cap  H_{1}.$$
We know that $T$ is an $(k-1,k,t)$-periodic subspace. For subspace $W$ of dimension $k-1$,  there must be
$  {\rm{proj}}_{[1,k]}(T) \cap H_{1}$ elements contained in the coset of $W \cap H_{1}$. Thus,
$$|  {\rm{proj}}_{[1,k]}(T) \cap H_{1}|  \le |W \cap H_1 | \le q^{\dim_q(W \cap  H_{1})}.$$

Similarly, for
$$f_{j} \in    {\rm{proj}}_{[(j-1)k+1,jk]}(T) \cap H_{j},$$
the number of its elements are
$$|  {\rm{proj}}_{[(j-1)k+1,jk]}(T) \cap  H_{j}| \le |W \cap H_j | \le q^{\dim_q(W \cap  H_{j})}  .$$ 
Then the number of elements of the set $S$ satisfies
\[|{S}| = \mathop  \prod \limits_{j = 1}^t |{\rm{pro}}{{\rm{j}}_{[(j - 1){k} + 1,j{k}]}}({T}) \cap {H_{j}}| \le {q^{\sum\limits_{j = 1}^t {\dim_q({W} \cap {H_{j}})} }}.\]
Since ${\cal H} = (H_{1},H_{2},\cdots, H_{t})$ are subspace designs with parameters $(k-1, N-d)_q$, according to Definition \ref{Df7}, $${\sum\limits_{j=1}^{t} {\dim_q(W \cap H_{j})}} \le N-d.$$ Hence, set $S$ has dimensions at most $N-d$ over $\mathbb{F}_q$.
\end{proof}

Based on the subspace design, the candidate messages of $[{\bf n}, k, d]_{{q^n}/q}$-sum-rank metric codes, when processed under the list decoding algorithm, are output in set $S$. Moreover, the dimension of these candidate messages over $\mathbb{F}_q$ is at most $N-d$. The list size of the constructed sum-rank metric code $\overline{\mathcal{C}}$ is derived from Theorem \ref{Th25} by applying the list decoding algorithm.





\begin{corollary}\label{Co26} 
Under the list decoding algorithm based on subspace design, the list size of MSRD code $\overline{\mathcal{C}}_1$ with the parameters $[(4|\cdots|4),2(2t-1),2]_{{q^2}/q}$ is at most $q^{4t-2}$.   
\end{corollary} 

\begin{corollary}\label{Co27} 
Under the list decoding algorithm based on subspace design, the list size of MSRD code $\overline{\mathcal{C}}_2$ with the parameters $[(8|\cdots|8),4(2t-1),2]_{{q^2}/q}$ is at most $q^{8t-2}$.  
\end{corollary}

\begin{corollary}\label{Co28}
Under the list decoding algorithm based on subspace design, the list size of MSRD code $\overline{\mathcal{C}}_3$ with the parameters $[(n | \cdots|n | 2n), 2, t(n-1)+1]_{{q^n}/q}$ is at most $q^{n+t-1}$.  
\end{corollary} 

\begin{corollary}\label{Co29} 
Under the list decoding algorithm based on subspace design, the list size of MSRD code $\overline{\mathcal{C}}_4$ with the parameters $[(n|\cdots|n|4n),4,t(n-3)+3]_{{q^n}/q}$ is at most $q^{3(n+t-1)}$.  
\end{corollary}

\begin{table}[H]
	\setlength{\abovecaptionskip}{0cm}
	\tabcolsep=0.04cm
	\renewcommand\arraystretch{1.2}
	\renewcommand{\tablename}{Table}
	\caption{The comparison of list sizes $L$}
	\addtocounter{table}{-1}
	\label{tab:lable}
	\begin{center}
		\begin{tabular}{ccccc}
			\toprule[1.5pt]
			\makecell{Sum-rank\\metric codes}     &   \makecell{Parameters}                &\makecell{ List size}              & \makecell{List size based on\\subspace designs}   \\
			\midrule[1.5pt]
			$\overline{\mathcal{C}}_1$ &  $[(4|\cdots|4),2(2t-1),2]_{{q^2}/q}$  &  $1+tq^{t(5+t)}$                   &   $q^{4t-2}$   \\
			$\overline{\mathcal{C}}_2$ &  $[(8|\cdots|8),4(2t-1),2]_{{q^2}/q}$  &  $1+tq^{t(9+t)}$     &  $q^{8t-2}$   \\
$\overline{\mathcal{C}}_3$ &  $[(n | \cdots|n | 2n), 2, t(n-1)+1]_{{q^n}/q}$   &  $1+(t-\lfloor {\frac{t(n-1)}{2}} \rfloor)q^{t(3n+t-1)}$    & $q^{n+t-1}$\\
$\overline{\mathcal{C}}_4$ &  $[(n|\cdots|n|4n),4,t(n-3)+3]_{{q^n}/q}$   &  $1+(t-\lfloor {\frac{t(n-3)}{2}+1} \rfloor)q^{t(5n+t-1)}$   & $q^{3(n+t-1)}$ \\
			\bottomrule[1.5pt]
		\end{tabular}
	\end{center}
\end{table}

For the sum-rank metric codes constructed in Section \ref{sec3}, Table 1 was generated by comparing the list sizes under the list decoding algorithms presented in Sections \ref{sec4} and \ref{sec3}, which illustrates that the list decoding algorithm based on subspace designs can enhance the decoding success rate.



\

\section{Conclusion}\label{sec5}
In this study, we constructed sum-rank metric codes under the actions of orthogonal spaces over finite fields. Moreover, the list size of the sum-rank metric codes under the list decoding algorithm is computed based on the subspace designs.

First, by leveraging the companion matrices of primitive polynomials, we obtain the cyclic orthogonal group of order $q^n -1$ and the Abelian non-cyclic orthogonal group of order ${(q^n -1)}^2$, and then construct MRD codes with parameters $(n \times {2n}, q^{2n}, n)_q$ and $(n \times {4n}, q^{4n} , n)_q$. Next, we present two methods for constructing sum-rank metric codes and obtain sum-rank metric codes with parameters $[(4|\cdots|4),2(2t-1),2]_{{q^2}/q}$, $[(8|\cdots|8),4(2t-1),2]_{{q^2}/q}$, $[(n | \cdots|n | 2n), 2, t(n-1)+1]_{{q^n}/q}$ and $[(n|\cdots|n|4n),4,t(n-3)+3]_{{q^n}/q}$.

We then connect the sum-rank metric codes with subspace designs through the $[{\bf{n}},k,d]_{{q^n}/q}$-system. Based on the constructed sum-rank metric codes, we calculated the list size for the geometric structure of the subspace design. Specifically, the list size of MSRD code with parameters $[(4|\cdots|4),2(2t-1),2]_{{q^2}/q}$ is at most $q^{4t-2}$; similarly, the list size of MSRD code with parameters $[(8|\cdots|8),4(2t-1),2]_{{q^2}/q}$ is at most $q^{8t-2} $; for the sum - rank metric codes with parameters $[(n | \cdots|n | 2n), 2, t(n-1)+1]_{{q^n}/q}$, the list size is at most $q^{n+t-1}$; and for the sum - rank metric codes with parameters $[(n | \cdots|n | 4n), 4, t(n-3)+3]_{{q^n}/q}$, the list size is at most $q^{ 3(n+t-1)}$. Finally, we compared the output list sizes under different list decoding algorithms, and proved that list decoding based on subspace designs can improve the decoding success rate.

\section*{Acknowledgement}
 X.-M Liu, J.-R Zhang and G. Wang are supported by the National Natural Science Foundation of China (No. 12301670), the Natural Science Foundation of Tianjin (No. 23JCQNJC00050), the Scientific Research Project of Tianjin Education Commission (No. 2023ZD041), the Fundamental Research Funds for the Central Universities of China (No. 3122024PT24) and the Graduate Student Research and Innovation Fund of Civil Aviation University of China (No. 2024YJSKC06001).

\end{document}